\documentclass[12pt, draftclsnofoot, onecolumn,romanappendices]{IEEEtran}

\usepackage{bm,cite}
\usepackage{amsmath}
\usepackage{amsthm}
\usepackage{amsbsy}
\usepackage{epsfig}
\usepackage{latexsym}
\usepackage{amssymb}
\usepackage{wasysym}
\usepackage{placeins}
\usepackage{color}
\usepackage{rotating}
\usepackage[usenames,dvipsnames]{xcolor}
\usepackage{dsfont}
\usepackage{caption}
\usepackage{subcaption}
\usepackage{algorithm}
\usepackage{algorithmic}

\usepackage{tikz}

\usetikzlibrary{shapes,arrows,snakes}
\usetikzlibrary{calc}

\IEEEoverridecommandlockouts
\DeclareMathAlphabet{\mathbit}{OML}{cmr}{bx}{it}
\DeclareMathAlphabet{\mathsf}{OT1}{cmss}{m}{n}
\DeclareMathAlphabet{\mathTXf}{OT1}{cmss}{bx}{it}

\DeclareMathOperator{\ZF}{ZF}

\DeclareMathOperator{\DoF}{DoF}

\DeclareMathOperator{\CN}{\mathcal{N}_{\mathbb{C}}}

\DeclareMathOperator{\QMAT}{\mathrm{Q-MAT}}

\newcommand{\bU}{\mathbf{U}}
\newcommand{\bV}{\mathbf{V}}

\newcommand{\bh}{\bm{h}}

\newcommand{\bbm}{\bm{m}}

\newcommand{\bu}{\bm{u}}
\newcommand{\bv}{\bm{v}}

\newcommand{\LB}{\left(}
\newcommand{\RB}{\right)}
\newcommand{\LSB}{\left[}
\newcommand{\RSB}{\right]}

\newcommand*{\dotleq}{\mathrel{\dot{\leq}}}

\newcommand{\E}{{\mathbb{E}}}

\newcommand{\He}{{{\mathrm{H}}}}

\newcommand{\xv}{\mathbf{x}}

\theoremstyle{remark}
\newtheorem{remark}{Remark} 
\theoremstyle{assumption}
\newtheorem{theorem}{Theorem}

\newtheorem{lemma}{Lemma}

\newtheorem{example}{Example}

\begin{document}
\title{Optimal DoF of the K-User Broadcast Channel with Delayed and Imperfect Current CSIT}

\author{\IEEEauthorblockN{Paul de Kerret, David Gesbert, Jingjing Zhang, and Petros Elia}

\thanks{Preliminary results will be presented at the IEEE Information Theory Workshop (ITW), 2016.}
\thanks{The authors are with the Communication Systems Department, EURECOM, Sophia Antipolis, France (email: \{dekerret,gesbert,zhang,elia\}@eurecom.fr).}
\thanks{D. Gesbert and P. de Kerret are supported by the European Research Council under the European Union's Horizon 2020 research and innovation program (Agreement no. 670896). P. Elia is supported by the ANR project ECOLOGICAL-BITS-AND-FLOPS.
}}

\author{\IEEEauthorblockN{Paul de Kerret, David Gesbert, Jingjing Zhang, and Petros Elia}\\
Communication Systems Department,\\ EURECOM}

\maketitle
\begin{abstract}
This work\footnote{D. Gesbert and P. de Kerret are supported by the European Research Council under the European Union's Horizon 2020 research and innovation program (Agreement no. 670896). P. Elia is supported by the ANR project ECOLOGICAL-BITS-AND-FLOPS.\\Parts of these results have been published in IEEE Information Theory Workshop 2016, Cambridge.} studies the optimal Degrees-of-Freedom (DoF) of the $K$-User MISO Broadcast Channel (BC) with delayed Channel-State Information at the Transmitter (CSIT) and with additional current noisy CSIT where the current channel estimation error scales in~$P^{-\alpha}$ for $\alpha\in[0,1]$. This papers establishes for the first time the optimal DoF in this setting thanks to a new transmission scheme which achieves the elusive DoF-optimal combining of the Maddah-Ali and Tse scheme (MAT) introduced in their seminal work in $2010$ with Zero-Forcing (ZF) for an arbitrary number of users. The derived sum DoF takes the surprisingly simple form $(1-\alpha) K/H_K+\alpha K$ where $H_K\triangleq \sum_{k=1}^K \frac{1}{k}$ is the sum-DoF achieved using solely MAT.
\end{abstract}

\section{Introduction}\label{se:SM}

In the $K$-user wireless BC, feedback accuracy and timeliness crucially affects performance, but are also notoriously difficult to obtain. In terms of accuracy, it is well known that increasing feedback quality can elevate performance, from that of TDMA (sum DoF of $1$), to the maximum possible interference-free performance with a sum DoF of $K$. As the recent result in \cite{Davoodi2016} tells us, having imperfect instantaneous CSIT with an estimation error that scales (in the high-power $P$ setting) as~$P^{-\alpha}$ ($\alpha\in[0,1]$), can allow, using basic ZF precoding techniques and rate splitting, for an optimal sum-DoF of $1+(K-1)\alpha$.

On the other hand, when perfect-accuracy CSIT is obtained in a delayed manner --- in the sense that the CSI is fed back with a delay exceeding the channel coherence period --- then, using more involved, retrospective, \emph{MAT}-type space-time alignment \cite{MaddahAli2012}, one can surprisingly get substantial DoF gains, reaching a sum-DoF of $K/H_K$ with $H_K\triangleq \sum_{k=1}^K \frac{1}{k}$, which scales with $K$ approximately as $K/\text{ln}(K)$.

This interplay between performance and feedback timeliness-and-quality, has sparked a plethora of works that considered a variety of feedback mechanisms with delayed and imperfect CSIT. Such works can be found in \cite{Vaze2012a,Abdoli2011,Lee2012}, and in \cite{Tandon2012b} which --- for the two-user MISO BC setting --- studied the case where the CSIT can alternate between perfect, delayed (completely outdated), and non-existent (see also~\cite{Rassouli2015}). The performance at finite SNR of the MAT scheme were also discussed in \cite{Yi2013a}, while the gap to optimality was bounded in \cite{Vahid2016}.

An interesting approach came with the work in \cite{Kobayashi2012} which introduced a feedback scenario that offered a combination of imperfect-quality current (instantaneously available) CSIT, together with additional (perfect-accuracy) delayed CSIT.
In this same setting --- which reflected different applications, including that of using predictions to get an estimate of the current state of a time-correlated channel --- the channel estimation error of the current channel state was assumed to scale in power as $P^{-\alpha}$, for some CSIT quality exponent $\alpha\geq 0$. \cite{Kobayashi2012} also introduced ingredients that proved to be key in this setting: The use of an initial imperfect ZF precoding layer, followed by retransmission of a \emph{quantized} form of the interference generated by CSI imperfections in the first layer precoder.

Additional work --- within the context of the BC --- came in \cite{Yang2013,Gou2012} which established the maximal DoF in a two-user MISO BC scenario, as well as in \cite{Chen2012a,Chen2013b} which considered the case of imperfect-quality delayed CSIT. More results can also be found in \cite{Chen2013a,Chen2013b} which considered the broad setting of any-time any-quality feedback, and in \cite{Yi2014,Chen2014} which studied the two-user MIMO BC (and IC); all for the two-user case.
This general challenge of dealing with imperfect feedback has also sparked very recent interest, with different publications that include \cite{Torrellas2015,Rassouli2016,Joudeh2016,Hao2016,Bracher2015,Lashgari2015}. Finally, interesting connections between the delayed CSIT configuration and the generalized feedback setting \cite{Shayevitz2013} has been put forward in \cite{Kim2015}.

\subsection{Simultaneous scaling of MAT and ZF gains}

For the more general case of the $K$-user BC, again with joint delayed and imperfect-current CSIT, very little is known. For the particular case considered here, a general outer bound was provided in \cite{dekerret2013_ISIT}, and efforts to reach this bound can be found in~\cite{Luo2014}. The main goal has remained to secure simultaneous scaling of MAT-type gains (that exploit delayed CSIT), and ZF gains (that exploit imperfect-quality current CSI). To date, this has remained an elusive open problem, and any instance of providing such simultaneous gains was either limited to the 2-user case, or --- as in the case of the scheme in \cite{dekerret2013_ISIT} --- resulted in MAT-type DoF gains that saturated at $2$. This elusive open problem is resolved here, by inventing a new scheme, referred to as the $\QMAT$ scheme, that combines different new ingredients that jointly allow for MAT and ZF components to optimally coexist. Combined with the outer bound in \cite{dekerret2013_ISIT}, the achieved DoF establishes the optimal sum-DoF, which is here shown to be equal to $\alpha K+ (1-\alpha) K/(\sum_{k=1}^K \frac{1}{k})$.

\subsection{Notation}
For $\mathcal{C}(P)$ denoting the sum capacity\cite{Cover2006} of the MISO BC considered, we will place emphasis on the high-SNR degree of freedom approximation 
\begin{equation}
\DoF^{\star}\triangleq \lim_{P\rightarrow \infty}\frac{\mathcal{C}(P)}{\log_2(P)}.
\label{eq:SM_4}
\end{equation}
We will use the notation $H_K \triangleq \sum_{i=1}^{K} \frac{1}{i}$ to represent the $K$-th harmonic number. $\mathbb{Z}$ will represent the integers, $\mathbb{Z}^{+}$ the positive integers, $\mathbb{R}^{+}$ the positive real numbers, $\mathbb{C}$ the complex numbers, $\binom{n}{k}$ the $n$-choose-$k$ operator, and $\oplus$ the bitwise XOR operation. We will use $\mathcal{K} \triangleq \{1,2,\cdots,K\}$. If $\mathcal{S}\subset \mathcal{K}$ is a set, then $\bar{\mathcal{S}}$ will denote $\mathcal{K} \backslash \mathcal{S}$, and $|\mathcal{S}|$ will denote its cardinality. Complex vectors will be denoted by lower-case bold font. For any vector $\xv$, we will use $\{\xv\}_i, \|\xv\|^2$ and $\xv^{\He}$ to respectively denote the $i$th element of the vector, its magnitude-squared, and its conjugate transpose. We will also use $\doteq$ to denote \emph{exponential equality}, i.e., we write $f(P)\doteq P^{B}$ to denote $\displaystyle\lim_{P\to\infty}\frac{\log(f(P))}{\log (P)}=B$. We write $\CN(0,\sigma^2)$ to denote the complex Gaussian distribution of zero mean and variance $\sigma^2$.

\subsection{System Model}\label{se:SM}

\subsubsection{$K$-User MISO Broadcast Channel}

This work considers the $K$-User MISO BC with fading, where the Transmitter (TX) --- which is equipped with $M$~antennas ($M\geq K$) --- serves $K$~single-antenna Receivers (RXs). At any time~$t$, the signal received at RX~$k\in\mathcal{K}$, can be written as
\begin{equation}
y_k[t]=\bm{h}_k^{\He}[t]\xv[t]+n_k[t]
\label{eq:SM_1}
\end{equation}
where $\bm{h}_k^{\He}[t]\in \mathbb{C}^{1\times M}$ represents the channel to user~$k$ at time~$t$, where~$\xv[t]\in \mathbb{C}^{M}$ represents the transmitted signal, and where $n_k[t]\in \mathbb{C}$ represents the additive noise at RX~$k$, where this noise is distributed as~$\CN(0,1)$, independently of the channel and of the transmitted signal. Furthermore, the transmitted signal~$\xv[t]$ fulfills the average asymptotic power constraint~$\E[\|\xv[t]\|^2]\doteq P$.
The channel is assumed to be drawn from a continuous ergodic distribution such that all the channel matrices and all their sub-matrices are almost surely full rank.

\subsubsection{Perfect delayed CSIT and imperfect current CSIT}

Our CSIT model builds on the delayed CSIT model introduced in~\cite{MaddahAli2012} and generalized in \cite{Kobayashi2012,Yang2013,Gou2012} to account for the availability of an imperfect estimate of the current channel state. For ease of exposition, we will here adopt the fast-fading channel model, and will assume that at any time~$t$, the TX has access to the delayed CSI (with perfect accuracy) of all previous channel realizations up to time~$t-1$, as well as an imperfect estimate of the current channel state. Each current estimate $\hat{\bm{h}}_k^{\He}[t]$ for each channel $\bm{h}_k^{\He}[t]$, comes with an estimation error
\begin{equation}
\tilde{\bm{h}}_k^{\He}[t] = \bm{h}_k^{\He}[t] - \hat{\bm{h}}_k^{\He}[t]
\label{eq:SM_2}
\end{equation}
whose entries are i.i.d. $\CN(0,P^{-\alpha})$ with power $P^{-\alpha}$ for some parameter $\alpha\in [0,1]$, which we refer to as the \emph{CSIT quality exponent}, and which is used to parameterize the accuracy of the current CSIT\footnote{Note that from a DoF perspective, we can restrict ourselves to $\alpha \in [0,1]$, since an estimation/quantization error with power scaling as~$P^{-1}$ ($\alpha = 1$), is essentially perfect. Similarly an estimation error with power scaling as~$P^{0}$ ($\alpha = 0$), offers no DoF gains over the case of having no CSIT (cf.~\cite{Davoodi2016}, see also~\cite{dekerret2013_ISIT}).}. All estimates are assumed to be independent of all estimation errors. Finally we make the common assumption that the channel~$\bm{h}_k^{\He}[t]$ is independent of all previous channel estimates and channel estimation errors, when conditioned on~$\hat{\bm{h}}_k^{\He}[t]$. We also adhere to the common convention (see~\cite{MaddahAli2012,Kobayashi2012,Gou2012}) of assuming perfect and global knowledge of channel state information at the receivers (perfect global CSIR), where the receivers know all channel states and all estimates.

\section{Main Results}\label{se:main}
We proceed directly with the main result.
\begin{theorem}
\label{thm:mainTheorem}
In the $K$-user MISO BC ($M\geq K$) with perfect delayed CSIT and $\alpha$-quality current CSIT, the optimal sum DoF is
\begin{equation}
\DoF^{\star}(\alpha)=(1-\alpha)  \frac{K}{\sum_{k=1}^K \frac{1}{k}}+\alpha K.
\label{eq:Main_1}
\end{equation}
\end{theorem}

\begin{figure} 
\centering
\includegraphics[width=0.7\columnwidth]{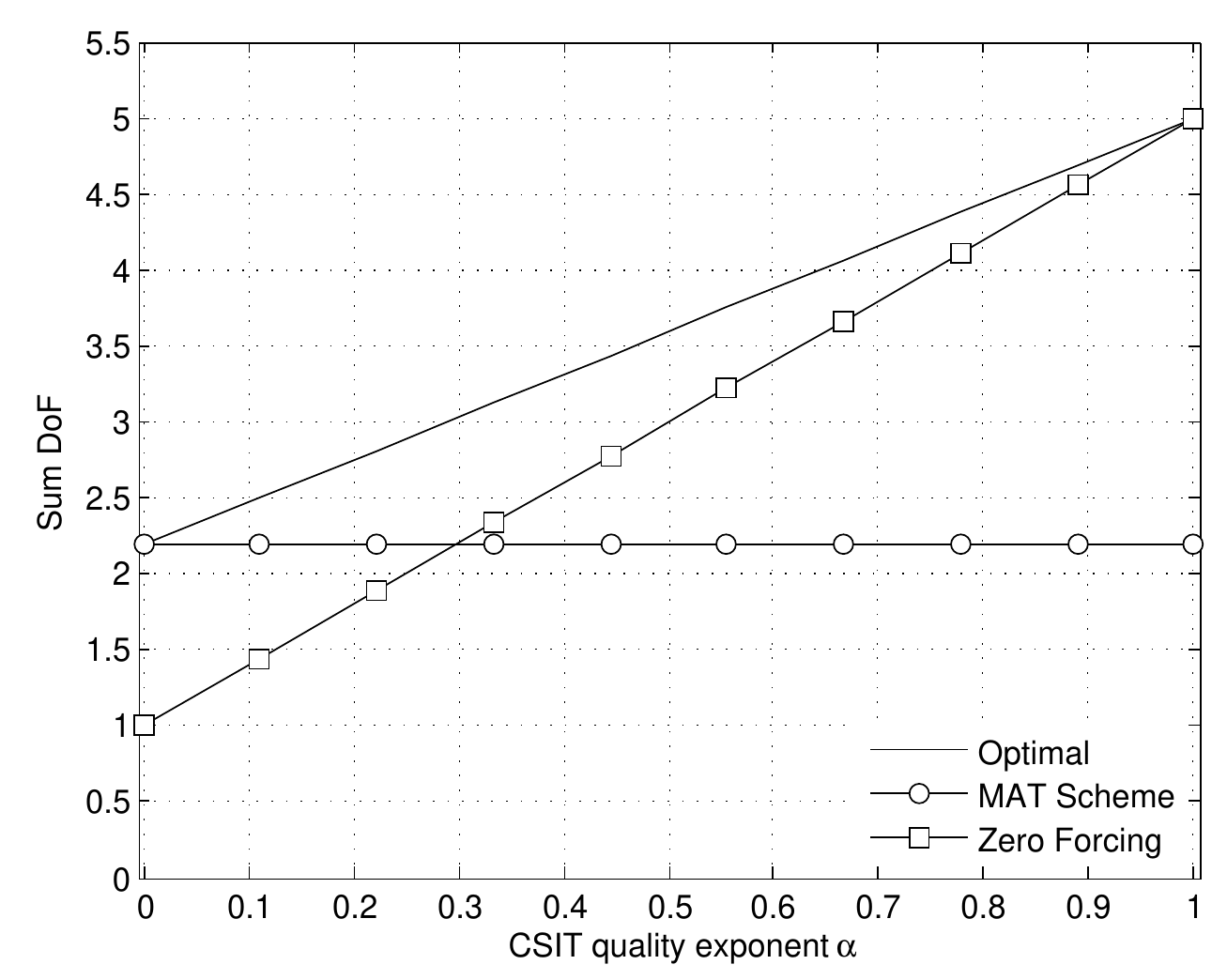}
\caption{Sum DoF achieved as a function of the CSIT scaling exponent~$\alpha$ for $K=5$~users.}
\label{ITW2016_journal_DoF_5users}
\end{figure}
\begin{proof}
The scheme that achieves this DoF, is described in the next sections. The optimality follows from the fact that this scheme's performance matches the DoF outer bound in~\cite{dekerret2013_ISIT}.
\end{proof}

Figure~\ref{ITW2016_journal_DoF_5users} shows the optimal DoF for the $5$-user case, and compares this performance to the MAT-only and ZF-only DoF performance. The proposed scheme significantly outperforms both ZF and MAT, and the gap increases with $K$. 

We proceed by describing first the simpler case of $K=2$ so as to get the proper insight to understand the general $K$ case, presented after. 
\FloatBarrier
\section{Scheme for the $K=2$ user case}\label{se:2users}
As stated above, we begin with the scheme description for the two-user case. Although the main motivation for the our scheme is clearly the case of $K>2$, this simple setting contains all the key components and allows to present in detail all the steps, with lighter notations than in the general $K$-user case. Note that the $K=3$-user case is described in the conference version of this work \cite{dekerret2016_ITW} and can serve as an additional example to gain insights before dealing with the general case \footnote{While a bit premature at this stage, we hasten to note that this exposition of the two-user case, nicely accentuates the difference between our scheme and previous efforts \cite{Gou2012,Yang2013} which --- albeit optimal for the two user case --- had to deviate from the MAT-type canonical structure in order to accommodate for the ZF component, thus making it difficult to extend to higher dimensions.}. 

\subsection{Encoding and Transmission}
The transmission follows closely the MAT multi-destination multi-layer scheme in \cite{MaddahAli2012} and we refer to this work for a nice and intuitive description of the MAT scheme. Consequently, our scheme is also divided into two phases for the $2$-user case, where phase~$1$ corresponds to the transmission of order-$1$ data symbols (meant for one user at a time) and spans $2$ Time Slots (TS), while phase~$2$ corresponds to the transmission of order-$2$ data symbols (meant for both users simultaneously) and spans $1$~TS.

A first deviation from existing schemes can be found in the fact that the scheme requires several \emph{rounds} (each of them following the structure in phases of the MAT scheme). 

\begin{remark}
The encoding across rounds was already used in \cite{dekerret2013_ISIT} and is a consequence of the delayed CSIT assumption. Indeed, the CSI necessary for the transmission of some of the data symbols (so-called \emph{auxiliary data symbols}) is not available when the transmission occurs. Therefore, this auxiliary data symbol is instead transmitted in the same phase of the next round.\qed
\end{remark}

We describe the transmission for an arbitrary round~$N$, and we will use the shorthand notation~$\bullet^{(R_N)}$ to denote the fact that the index of the round is equal to~$N$. The particularities of the first and last rounds will be clarified later on.
\subsubsection{Phase~$1$}
Phase~$1$ spans two time slots, denoted by $t_1^{(R_N)}$ and $t_2^{(R_N)}$.
For $t = t_1^{(R_N)}$, the transmitted signal is given by
\begin{equation}
\xv[t]\!=\!\bV[t]\bbm[t]\!+ \bv_{2}^{\ZF}[t]a_2[t]\!+\!\sum_{k=1}^{2}\bv_{k}^{\ZF}[t]s_k[t]
\label{eq:2user_1}
\end{equation}
where
\begin{itemize}
\item $\bbm[t] \in \mathbb{C}^{2}$ is a vector containing two so-called \emph{Q-MAT data symbols} meant for user~$1$, each carrying $(1-\alpha)\log_2(P)$~bits, where the first symbol is allocated full power $\E\LSB|\{\bbm[t]\}_1 |^2\RSB\doteq P$, while the second symbol is allocated lesser power $\E\LSB|\{\bbm[t]\}_2 |^2\RSB\doteq P^{1-\alpha}$. Furthermore $\bV[t]\in\mathbb{C}^{M\times 2}$ is defined as
\begin{equation}
\bV[t]\triangleq \begin{bmatrix} \bv_{1}^{\ZF}[t]&\bu_1\end{bmatrix}
\label{eq:2user_2}
\end{equation}
where $\bv_{1}^{\ZF}[t] \in \mathbb{C}^{M}$ is the unit-norm ZF beamformer aimed at user~$1$ (i.e., which is orthogonal to the current estimate of the channel to user 2), while $\bu_1\in \mathbb{C}^{M}$ is a unit-norm vector that is randomly drawn and isotropically distributed.
\item $a_2[t] \in \mathbb{C}$ is a so-called \emph{auxiliary data symbol} meant for user~$2$, carrying $\min(1-\alpha,\alpha)\log_2(P)$~bits (generally from previous interfering terms), and allocated full power $\E\LSB|a_2[t]|^2\RSB\doteq P$.
\item $s_k[t]\in \mathbb{C}, \ k\in \{1,2\}$ are \emph{ZF data symbols} meant for user~$k$, each carrying $\alpha\log_2(P)$~bits and each having power $\E\LSB|s_k[t]|^2\RSB\doteq P^{\alpha}$.
\end{itemize}

Upon omitting the noise realizations for simplicity, the received symbols during $t=t_1^{(R_N)}$, can be written as
\begin{equation}
\begin{aligned}
y_{1}[t]&\!=\!\underbrace{\bh_1^{\He}[t]\bV[t]\bbm[t]}_{\doteq P} +\underbrace{\bh_1^{\He}[t]\bv_2^{\ZF}[t] a_2[t]}_{\doteq P^{1-\alpha}}\!+\!\underbrace{z_1[t]}_{\doteq P^{\alpha}}\\
y_2[t]&=\underbrace{\bh_2^{\He}[t]\bv_2^{\ZF}[t] a_2[t]}_{\doteq P}+ \underbrace{i_2[t]}_{\doteq P^{1-\alpha}} +\underbrace{z_2[t]}_{\doteq P^{\alpha}}
\end{aligned}
\label{eq:2user_3}
\end{equation}
where
\begin{equation}
\begin{aligned}
i_2[t]&\triangleq\underbrace{\bh_2^{\He}[t]\bV[t]\bbm[t]}_{\doteq P^{1-\alpha}}
\end{aligned}
\label{eq:2user_4}
\end{equation}
and where for any $t\in \mathbb{Z}^{+}$,
\begin{equation}
\begin{aligned}
\!\!z_k[t]&\!\triangleq\! \underbrace{\bh_k^{\He}[t]\bv_k^{\ZF}[t] s_k[t]\!}_{\doteq P^{\alpha}}+\underbrace{\bh_{k}^{\He}[t]\bv_{\bar{k}}^{\ZF}[t] s_{\bar{k}}[t]}_{\doteq P^{0}}, \qquad k\in \{1,2\}.
\end{aligned}
\label{eq:2user_5}
\end{equation}
In the above equations, underneath each summand, we describe the asymptotic approximation of the power of the corresponding term.

For $t= t_2^{(R_N)}$, the transmission is described by the above equations, after exchanging the indices of the two users ($1 \leftrightarrow 2$). Thus, the received signals during $t=t_2^{(R_N)}$, take the form
\begin{equation}
\begin{aligned}
y_1[t]&\!=\underbrace{\bh_1^{\He}[t]\bv_1^{\ZF}[t] a_1[t]}_{\doteq P}\!+\underbrace{i_1[t]}_{\doteq P^{1-\alpha}} +\!\underbrace{z_1[t]}_{\doteq P^{\alpha}}\\
y_2[t]&=\underbrace{\bh_2^{\He}[t]\bV[t]\bbm[t]}_{\doteq P}+\underbrace{\bh_2^{\He}[t]\bv_1^{\ZF}[t] a_1[t]}_{\doteq P^{1-\alpha}} +\underbrace{z_2[t]}_{\doteq P^{\alpha}}
\end{aligned}
\label{eq:2user_6}
\end{equation}
where now
\begin{equation}
\begin{aligned}
i_1[t]&\triangleq\underbrace{\bh_1^{\He}[t]\bV[t]\bbm[t]}_{\doteq P^{1-\alpha}}.
\end{aligned}
\label{eq:2user_7}
\end{equation}

\paragraph{Interference quantization}

At the end of phase~$1$, the TX can use its delayed CSIT to compute the interference terms $i_2[t_{1}^{(R_N)}]$ and $i_1[t_{2}^{(R_N)}]$. These terms are first quantized according to a specifically designed quantizer, before being used to generate the data symbols which will be transmitted either in the next phase of the same round or in the same phase of the next round. 

We now present the quantization scheme, which needs to be differentiated depending whether $\alpha \leq \frac{1}{2}$ or not. Let us first consider the more involved case of $\alpha \leq \frac{1}{2}$. Our quantizer is defined in the following lemma.
\begin{lemma}
\label{lemma_quantizer}
Let $Y$ be a unit-variance zero-mean random variable of bounded density~$p_Y$. Then, there exists a quantizer of rate $(\beta_1-\beta_2)\log_2(P)$~bits for any $0< \beta_2 \leq \beta_1$, denoted by $Q_{\beta_1,\beta_2}$, such that, for any unit-variance zero-mean random variable $n$, it holds that 
\begin{equation}
\lim_{P\rightarrow \infty}\Pr\{Q_{\beta_1,\beta_2}(\sqrt{P^{\beta_1}}y+\sqrt{P^{\beta_2}}n) = Q_{\beta_1,\beta_2}(\sqrt{P^{\beta_1}}y)\} = 1
\label{eq:2user_8}
\end{equation}
and at the same time
\begin{equation}
\E\LSB |Q_{\beta_1,\beta_2}(\sqrt{P^{\beta_1}}y)-\sqrt{P^{\beta_1}}y|^2\RSB\dotleq P^{\beta_2}.
\label{eq:2user_9}
\end{equation}
\end{lemma}
\begin{proof}
The design of the quantizer and the detailed proof of the result is given in the appendix.
\end{proof}
To convey some intuition on the particular properties of this quantizer, we now describe how it could be used in the following toy-example.
\begin{example}
Let us consider a setting where a first node, node~$A$, has the knowledge of a Gaussian random variable~$X\sim \CN(0,P)$ while another node, node~$b$ only obtains a corrupted version~$Y$ given by
\begin{equation}
Y=\underbrace{X}_{\doteq P}+\underbrace{N}_{\doteq P^{\gamma}}
\end{equation}
where $N$ is a Gaussian random variable with variance $P^{\gamma}$. Node~$A$ wants to transmit the minimum number of bits to node~$B$ in order for this node to reconstruct the random variable~$X$ up to the noise floor, as the power~$P$ increases. As a solution to this problem, we propose to use the quantizer~$Q_{1,\gamma}$ described in the previous lemma. Indeed, using this quantizer, it holds that
\begin{equation}
\lim_{P\rightarrow \infty}\Pr\{Q_{1,\gamma}(X+N) = Q_{1,\gamma}(X)\} = 1
\label{eq:2user_8}
\end{equation}
where 
\begin{equation}
X=Q_{1,\gamma}(X)+n_Q
\end{equation}
and with $n_Q\doteq P^{\gamma}$. Hence, if node~$A$ quantizes $n_Q$ using $\gamma\log_2(P)$~bits using an adequate quantizer from the literature, and transmit it to node~$B$, it follows from well-known results from Rate-Distortion theory that node~$B$ will be able to reconstruct~$X$ up to the noise floor, as the power~$P$ increases. Indeed, $Q_{1,\gamma}(X)$ can be obtained with probability one from~$X+N$, as $P$ increases. \qed
\end{example}

Let us now consider the quantization of the interference term $i_2[t_{1}^{(R_N)}]$ (scaling in $P^{1-\alpha}$). Using the above quantizer $Q_{\beta_1,\beta_2}$, and setting $\beta_1 = 1-\alpha, \beta_2 = \alpha$, offers us a quantized version $Q_{1-\alpha,\alpha}(i_2[t_{1}^{(R_N)}])$ that carries $(1-2\alpha)\log_2(P)$~bits, and which guarantees that the resulting quantization noise, which we denote by $n_2[t_{1}^{(R_N)}]$, has a power that scales as $P^{\alpha}$. The aforementioned quantization noise $n_2[t_{1}^{(R_N)}]$ is then itself re-quantized using any standard optimal quantizer with $\alpha\log_2(P)$~bits, which is known \cite{Cover2006} to guarantee quantization noise (from the second quantization) that scales as $P^0$. For $\hat{n}_2[t_{1}^{(R_N)}]$ denoting the quantized version of $n_2[t_{1}^{(R_N)}]$, we get the final combined estimate carrying $(1-\alpha)\log_2(P)$~bits in the form
\begin{equation}
\hat{i}_2[t_{1}^{(R_N)}]\triangleq Q_{1-\alpha,\alpha}(i_2[t_{1}^{(R_N)}]) + \hat{n}_2[t_{1}^{(R_N)}]
\label{eq:2user_10}
\end{equation}
where in the above the addition is over the complex numbers. This $2$-step quantization is illustrated in Fig.~\ref{ITW2016_journal_quantizer}.

For the easier case where $\alpha \geq \frac{1}{2}$, the interference term $i_1[t_{2}^{(R_N)}]$ is simply quantized using $(1-\alpha)\log_2(P)$ bits using a standard quantizer, guaranteed (cf.~\cite{Cover2006}) to have quantization noise that scales in $P^0$. The quantized signal obtained is also denoted by $\hat{i}_2[t_{1}^{(R_N)}]$. 
\begin{figure}
\centering
\includegraphics[width=0.7\columnwidth]{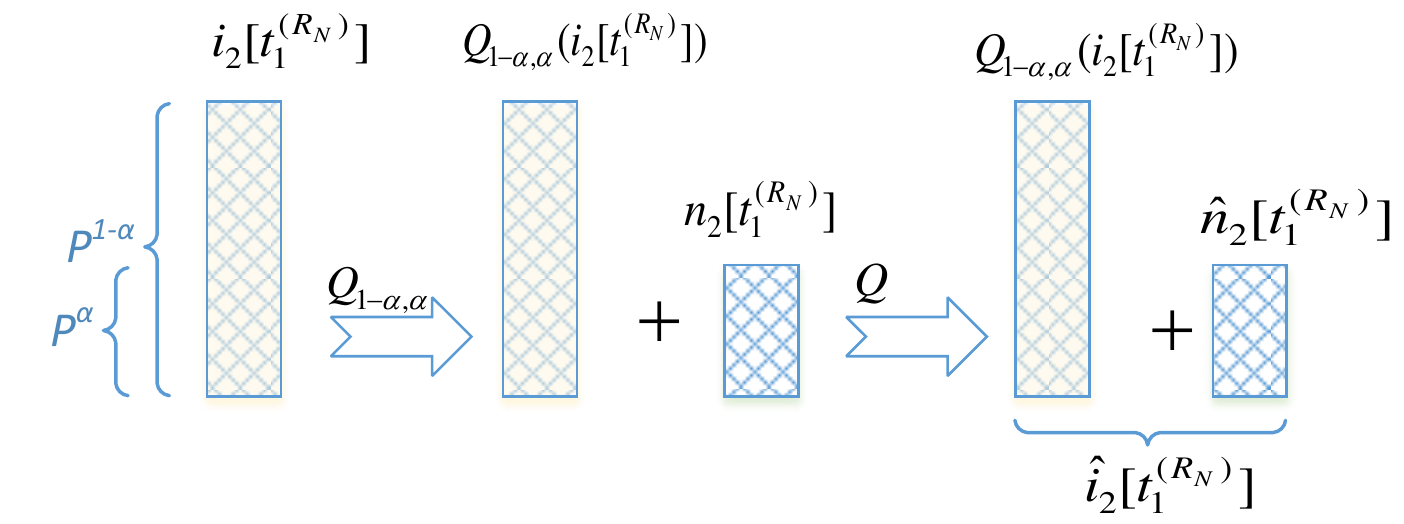}
\caption{Illustration of the $2$-step quantization scheme with the quantization noise at the noise floor being omitted.} 
\label{ITW2016_journal_quantizer}
\end{figure} 

The same quantization process is applied to $i_1[t_{2}^{(R_N)}]$ to obtain $\hat{i}_1[t_{2}^{(R_N)}]$ with the same rate and quantization-noise properties. These quantized bits will be placed in the auxiliary data symbols of the next round, as we now describe.

\paragraph{Generation of auxiliary data symbols for phase~$1$ of round $N+1$}
The auxiliary data symbol $a_{2}[t_{1}^{(R_{N+1})}]$ will carry
\begin{equation}
\begin{cases}
a_{2}[t_{1}^{(R_{N+1})}] \longleftarrow \hat{n}_2[t_{1}^{(R_N)}] &,\alpha\leq \frac{1}{2} \\
a_{2}[t_{1}^{(R_{N+1})}]\longleftarrow \hat{i}_2[t_{1}^{(R_N)}] &,\alpha\geq \frac{1}{2}
\end{cases}
\label{eq:2user_11}
\end{equation}
and similarly $a_{1}[t_{2}^{(R_{N+1})}]$ will carry the following quantization bits
\begin{equation}
\begin{cases}
a_{1}[t_{2}^{(R_{N+1})}]\longleftarrow \hat{n}_1[t_{2}^{(R_N)}] &,\alpha\leq \frac{1}{2} \\
a_{1}[t_{2}^{(R_{N+1})}]\longleftarrow \hat{i}_1[t_{2}^{(R_N)}] &,\alpha\geq \frac{1}{2}.
\end{cases}
\label{eq:2user_12}
\end{equation}
During the first round, all these auxiliary symbols are initialized to zero.

\paragraph{Generation of $\QMAT$ data symbols for phase~$2$ of round $N$}

During the second phase of round $N$, using delayed CSIT, the TX generates the following order-$2$ $\QMAT$ data symbol
\begin{align}
m[t_{1,2}^{(R_N)}]& \longleftarrow (\hat{i}_2[t_1^{(R_N)}] \oplus \hat{i}_1[t_2^{(R_N)}])
\label{eq:2user_13}
\end{align}
where $t_{1,2}^{(R_N)}$ denotes the TS used to transmit the order-$2$ data symbols.

\subsubsection{Phase~$2$}
Phase~$2$ consists of one TS $t=t_{1,2}^{(R_N)}$, during which
\begin{equation}
\xv[t]\!=\bv[t]m[t]+\!\sum_{k=1}^{2}\bv_{k}^{\ZF}[t]s_k[t]
\label{eq:2user_14}
\end{equation}
where
\begin{itemize}
\item $m[t] \in \mathbb{C}$ is an order-$2$ $\QMAT$ data symbol which carries $(1-\alpha)\log_2(P)$~bits that originate from the previous phase of the same round (see \eqref{eq:2user_13}), and which is destined for both users. The symbol is allocated full power $\E\LSB|m[t]|^2\RSB\doteq P$.
\item $s_k[t], \ k\in \{1,2\}$ is a ZF data symbol destined for user~$k$, carrying $\alpha\log_2(P)$~bits, and having power $\E\LSB|s_k[t]|^2\RSB\doteq P^{\alpha}$.
\end{itemize}
Upon omitting the noise realizations, the received signals during TS $ t=t_{1,2}^{(R_N)}$ take the form
\begin{equation}
\begin{aligned}
y_{1}[t]&=\underbrace{\bh_1^{\He}[t]\bv[t] m[t]}_{\doteq P}  +\underbrace{z_1[t]}_{\doteq P^{\alpha}}\\
y_2[t]&=\underbrace{\bh_2^{\He}[t]\bv[t] m[t]}_{\doteq P} +\underbrace{z_2[t]}_{\doteq P^{\alpha}} .
\end{aligned}
\label{eq:2user_15}
\end{equation}

\subsection{Decoding (achievability proof by induction)}\label{se:2users_decoding:induction}
We now turn to the decoding part, which here --- for the sake of clarity--- will be assumed to start after the end of transmission in all rounds and all phases. We will show that each user can decode all its desired data symbols. The proof has to be done by induction due to the fact that the auxiliary data symbols contain information coming from the previous round.

Let us consider without loss of generality the decoding at user~$1$. Our induction statement is that if the auxiliary data symbol $a_{1}[t_{2}^{(R_N)}]$ and $a_{2}[t_{1}^{(R_{N})}]$ are decoded at user~$1$, user~$1$ can decode:
\begin{itemize}
\item Its destined $\QMAT$ data symbols $\bbm[t_{1}^{(R_N)}]$ and $m[t_{1,2}^{(R_N)}]$
\item Its destined ZF data symbols $s_1[t_{1}^{(R_N)}],s_1[t_{2}^{(R_N)}],s_1[t_{1,2}^{(R_N)}]$
\item The auxiliary data symbols of the following round $a_{1}[t_{2}^{(R_{N+1})}]$ and $a_{2}[t_{1}^{(R_{N+1})}]$, thus allowing for the inductive step.
\end{itemize}
The initialization is done for $n=0$ by considering that all data symbols of round~$n=0$ have zero rate and by setting all auxiliary data symbols of round~$n=1$ to zero. Indeed, all data symbols for $n=0$ are decoded, as well as the auxiliary data symbols of round $n+1=1$. Thus, the induction property is satisfied for $n=0$. 

We then proceed to consider an arbitrary round~$N>0$. As part of the induction, we consider $a_{1}[t_{2}^{(R_N)}]$ and $a_{2}[t_{1}^{(R_N)}]$ to be already decoded at user~$1$.

\paragraph{Decoding of phase~$2$}
As a first step, the $\QMAT$ and ZF data symbols of phase~$2$ of round~$N$ are decoded using successive decoding. Indeed, as we see from \eqref{eq:2user_15}, the SINR of the $\QMAT$ data symbol~$m[t_{1,2}^{(R_N)}]$ is in the order of $P^{1-\alpha}$, which matches the scaling of the data symbol's rate. Using successive decoding, the ZF data symbol~$s_1[t_{1,2}^{(R_N)}]$ is then also decoded.

\paragraph{Decoding the interference in phase~$1$ (round $N$)}\label{se:2users_decoding:interference}
Receiver~$1$ then uses the signal received during TS~$t_{2}^{(R_{N+1})}$ (i.e., round $N+1$) to decode $a_1[t_{2}^{(R_{N+1})}]$. This is possible because, as seen in~\eqref{eq:2user_6}, the scaling of the SINR of $a_1[t_{2}^{(R_{N+1})}]$ is $P^{\min(\alpha,1-\alpha)}$. The content of $a_1[t_{2}^{(R_{N+1})}]$ depends on whether $\alpha\geq \frac{1}{2}$ or $\alpha\leq \frac{1}{2}$.
\begin{itemize}
\item If $\alpha\geq \frac{1}{2}$, user~$1$ has obtained $\hat{i}_1[t_{2}^{(R_{N})}]$ from $a_1[t_{2}^{(R_{N+1})}]$.
\item If $\alpha\leq \frac{1}{2}$, user~$1$ has obtained $\hat{n}_1[t_{2}^{(R_{N})}]$ from $a_1[t_{2}^{(R_{N+1})}]$. To recover the quantized interference $\hat{i}_1[t_{2}^{(R_{N})}]$, it is necessary for user~$1$ to also decode $Q_{1-\alpha,\alpha}\LB i_1[t_{2}^{(R_{N})}]\RB $ to form $\hat{i}_1[t_{2}^{(R_{N})}]$ in a similar way as in \eqref{eq:2user_10}. 

The term $Q_{1-\alpha,\alpha}\LB i_1[t_{2}^{(R_{N})}]\RB $ can be obtained at user~$1$ by applying the quantizer~$Q_{1-\alpha,\alpha}$ to the received signal~$y_1[t_{2}^{(R_{N})}]$. Indeed, following Lemma~\ref{lemma_quantizer}, it holds that
\begin{equation}
\lim_{P\rightarrow \infty} \Pr\left\{Q_{1-\alpha,\alpha}\LB y_1[t_{2}^{(R_{N})}]\RB=Q_{1-\alpha,\alpha}\LB i_1[t_{2}^{(R_{N})}]\RB\right\}=1.
\label{eq:2user_16}
\end{equation} 
Thus, in the limit of large $P$, user~$1$ can decode $Q_{1-\alpha,\alpha}\LB y_1[t_{2}^{(R_N)}]\RB$ with probability one, and then combine it with $\hat{n}_1[t_{2}^{(R_{N})}]$ to obtain $\hat{i}_1[t_{2}^{(R_{N})}]$.
\end{itemize}
In both cases $\alpha\geq \frac{1}{2}$ and $\alpha\leq \frac{1}{2}$, user~$1$ has obtained~$\hat{i}_1[t_{2}^{(R_{N})}]$.

\paragraph{Decoding of the $\QMAT$ data symbols of phase~$1$}
By induction, the auxiliary data symbol $a_1[t_{2}^{(R_{N})}]$ is known at user~$1$ such that its contribution to the received signal can be removed. Furthermore, due to the decoding of phase~$2$, user~$1$ has obtained
\begin{equation}
\begin{aligned}
m[t_{1,2}^{(R_N)}]=(\hat{i}_1[t_{2}^{(R_{N})}] \oplus \hat{i}_2[t_{1}^{(R_{N})}]).
\end{aligned}
\label{eq:2user_17}
\end{equation}
Using~$\hat{i}_1[t_{2}^{(R_N)}]$, user~$1$ obtains~$\hat{i}_2[t_{1}^{(R_N)}]$, and thus user~$1$ has knowledge (up to the noise level) of the following two components
\begin{equation}
\begin{aligned}
&\underbrace{\bh_1^{\He}[t_{1}^{(R_N)}]\bV[t_{1}^{(R_N)}]\bbm[t_{1}^{(R_N)}]+\bh_{1}^{\He}[t_{1}^{(R_N)}]\bv_2^{\ZF}[t_{1}^{(R_N)}] a_2[t_{1}^{(R_N)}]}_{\doteq P} +\underbrace{z_1[t_{1}^{(R_N)}]}_{\doteq P^{\alpha}}\\
&\underbrace{\bh_2^{\He}[t_{1}^{(R_N)}]\bV[t_{1}^{(R_N)}]\bbm[t_{1}^{(R_N)}]}_{\doteq P^{1-\alpha}}.
\end{aligned}
\label{eq:2user_18}
\end{equation}

By induction, $ a_2[t_{1}^{(R_N)}]$ is assumed to be already decoded at user~$1$, such that its contribution to the received signal in \eqref{eq:2user_18} can be removed. Consequently, user~$1$ has obtained two signals with a SINR scaling in $P^{1-\alpha}$, and can decode its two destined data symbols contained in $\bbm[t_{1}^{(R_N)}]$.

\paragraph{Decoding of the auxiliary data symbols of round~$N+1$}

The decoding of~$a_{2}[t_{1}^{(R_{N+1})}]$ follows directly from the definition of the auxiliary data symbol in \eqref{eq:2user_11}. Indeed, the auxiliary data symbol $a_{2}[t_{1}^{(R_{N+1})}]$ is a function of the $\QMAT$ data symbols~$\bbm[t_{1}^{(R_N)}]$, which have already been decoded by user~$1$. It can thus decode $a_{2}[t_{1}^{(R_{N+1})}]$, which provides the induction to the next round.

\paragraph{Decoding of the ZF data symbols at user~$1$}
User~$1$ has decoded its $\QMAT$ data symbols and the auxiliary data symbols needed for the next round. Thus, it remains to show that all the private data symbols destined to user~$1$ can also be decoded. For $t=t_1^{(R_N)}$ and $t=t_{1,2}^{(R_N)}$, user~$1$ has decoded the $\QMAT$ data symbols transmitted. Consequently, user~$1$ can use successive decoding to decode its destined ZF data symbol~$s_1[t_{1}^{(R_N)}]$ and $s_1[t_{1,2}^{(R_N)}]$ (see \eqref{eq:2user_3} and \eqref{eq:2user_15}). Yet, the $\QMAT$ data symbols sent during $t_{2}^{(R_N)}$ have not been decoded at user~$1$ (as destined to user~$2$), and thus successive decoding can not be used directly here. Yet, as user~$1$ has reconstructed $i_1[t_{2}^{(R_{N})}]$ (up to bounded noise), it can use it to remove the interference generated by the $\QMAT$ data symbols that were meant for user~$2$, and can consequently decode its own ZF data symbols (see~\eqref{eq:2user_6}).

\begin{remark}
The fact that user~$1$ is able to remove $i_1[t_{2}^{(R_{N})}]$ up to the noise floor, is the key property of our transmission scheme. This allows to remove all the received signals generated by the $\QMAT$ data symbols during all TS, thus making possible the ZF transmission with power $P^{\alpha}$.\qed
\end{remark}
The decoding process is illustrated in Fig.~\ref{ITW2016_journal_decoding}.
\begin{figure}
\centering
\includegraphics[width=0.7\columnwidth]{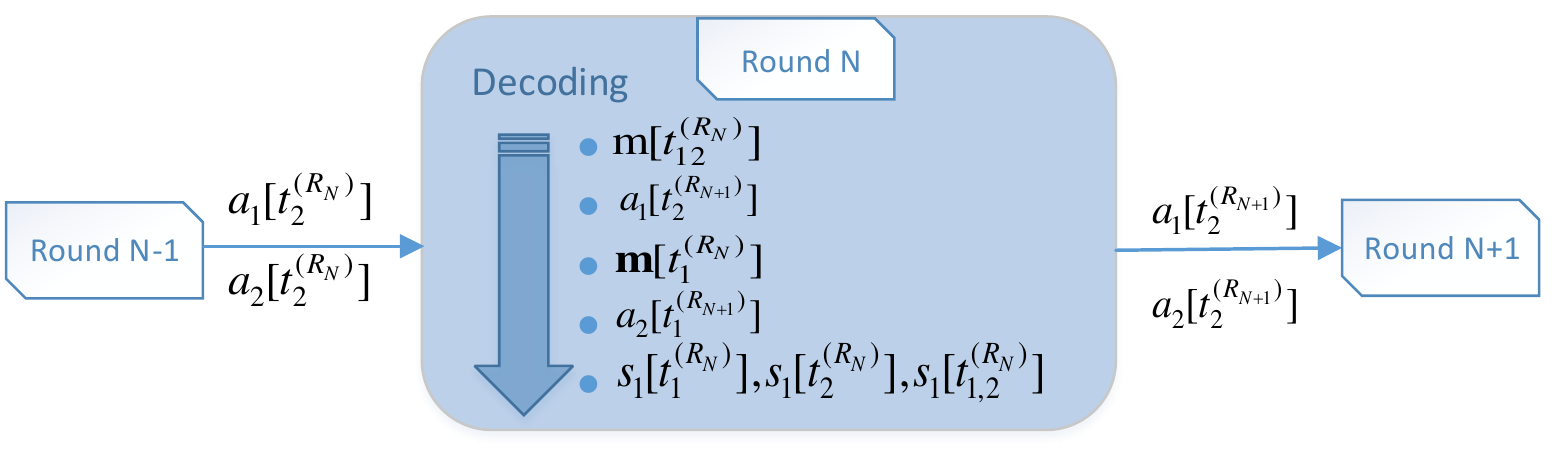}
\caption{Illustration of the decoding process for $K=2$.} 
\label{ITW2016_journal_decoding}
\end{figure}

\subsection{Calculation of the DoF}
In this particular $K=2$ setting, one $\QMAT$ round lasts $3$ TS during which $2$~$\QMAT$ data symbols of rate $(1-\alpha)\log_2(P)$~bits are transmitted to each user. Furthermore, one ZF data symbol of rate $\alpha\log_2(P)$~bits is transmitted to each user during each TS. In the last round, termination is achieved by sending only auxiliary symbols, which induces a small loss in DoF performance. This loss is made negligible by considering a large number of rounds. The resulting DoF is thus
\begin{equation}
\DoF^{\QMAT}=\frac{4(1-\alpha)+6\alpha}{3},
\label{eq:2user_19}
\end{equation}
which concludes the proof for the $2$-user case.

\section{$K$-user case: Encoding and Transmission}\label{se:K_user_encoding}
We now describe the $\QMAT$ scheme for an arbitrary number of users. The transmission spans several so-called $\QMAT$~rounds with each round following the multi-phase structure of MAT~\cite{MaddahAli2012} with $K$ phases. Phase $j$ aims to communicate order-$j$ (i.e., destined to $j$~users) data symbols. For clarity, our description will trace the description in Section~III.C of~\cite{MaddahAli2012}. In contrast to the original MAT scheme, the order-$j$ data symbols are now digitally encoded such that they are expected to be decoded at all $j$~corresponding RXs.

We proceed to describe the transmission during round~$N$, and assume that the transmissions up to round~$N-1$ have already been concluded. 

\subsection{Phase~$j$ for $j\in \{1,\ldots,K-1\}$}

Phase~$j$ consists in the transmission of $(K-j+1)\binom{K}{j}$~order-$j$ data symbols of rate $(1-\alpha)\log_2(P)$~bits, and in the process it leads to the generation of $j\binom{K}{j+1}$ order-$(j+1)$ data symbols of rate $(1-\alpha)\log_2(P)$~bits, as well as auxiliary symbols to be sent during phase~$j$ of round~$N+1$.

Phase $j$ of round $N$ spans $\binom{K}{j}$ time-slots, each dedicated to a subset $\mathcal{S}$ of users, for all $\mathcal{S}\subset \mathcal{K}$ of size $|\mathcal{S}|=j$ and where we denote by $t_{\mathcal{S}}^{(R_N)}$ the TS dedicated to~$\mathcal{S}$.

\subsubsection{Transmission at phase~$j$}

During $t= t_{\mathcal{S}}^{(R_N)}$, the transmit signal is given by
\begin{equation}
\xv[t_{\mathcal{S}}^{(R_N)}]\!=\!\bV[t_{\mathcal{S}}^{(R_N)}]\bbm[t_{\mathcal{S}}^{(R_N)}]\!+\sum_{\ell\in\bar{\mathcal{S}}}\bv^{\ZF}_{\ell}[t_{\mathcal{S}}^{(R_N)}] a_{\ell}[t_{\mathcal{S}}^{(R_N)}]+\!\sum_{k=1}^{K}\bv_{k}^{\ZF}[t_{\mathcal{S}}^{(R_N)}]s_k[t_{\mathcal{S}}^{(R_N)}]
\label{eq:Kuser_1}
\end{equation}
where
\begin{itemize}
\item $\bbm[t_{\mathcal{S}}^{(R_N)}] \in \mathbb{C}^{K-j+1}$ is a vector containing $(K-j+1)$ $\QMAT$ order-$j$ data symbols meant for the users in $\mathcal{S}$. Each symbol carries $(1-\alpha)\log_2(P)$~bits. The first symbol has full power $\E\LSB|\{\bbm[t_\mathcal{S}^{(R_N)}]\}_1 |^2\RSB\doteq P$, while the others have power $\E\LSB|\{\bbm[t_\mathcal{S}^{(R_N)}]\}_i |^2\RSB\doteq P^{1-\alpha},\forall i\in\{2,\ldots,K-j+1\}$. Furthermore, $\bV[t_\mathcal{S}^{(R_N)}]\in\mathbb{C}^{K\times (K-j+1)}$ is defined as
\begin{equation}
\bV[t_\mathcal{S}^{(R_N)}]\triangleq \begin{bmatrix} \bv_{\mathcal{S}}^{\ZF}[t_\mathcal{S}^{(R_N)}]&\bU_j\end{bmatrix}
\label{eq:Kuser_2}
\end{equation}
where $\bv_{\mathcal{S}}^{\ZF}[t_\mathcal{S}^{(R_N)}] \in \mathbb{C}^{K}$ is a unit-norm ZF precoder that is orthogonal to all current-CSI estimates of the channels of the users in~$\bar{\mathcal{S}}$, and where $\bU_{j}\in \mathbb{C}^{K\times (K-j)}$ is a randomly chosen, isotropically distributed unitary matrix.
\item $a_k[t_\mathcal{S}^{(R_N)}] \in \mathbb{C}, \ k\in \bar{\mathcal{S}}$ is an auxiliary data symbol meant for user~$k$, having rate $\min(1-\alpha,\alpha)\log_2(P)$~bits and power $\E\LSB|a_k[t_\mathcal{S}^{(R_N)}]|^2\RSB\doteq P$.
\item $s_k[t_\mathcal{S}^{(R_N)}], \ k\in \mathcal{K}$ is a ZF data symbol meant for user~$k$, having rate $\alpha\log_2(P)$~bits and power $\E\LSB|s_k[t_\mathcal{S}^{(R_N)}]|^2\RSB\doteq P^{\alpha}$.
\end{itemize}
At user $k\in\mathcal{S}$, the received signal for $t=t_{\mathcal{S}}^{(R_N)}$ then takes the form
\begin{align}
y_k[t_{\mathcal{S}}^{(R_N)}]&=\underbrace{\bh_k^{\He}[t_{\mathcal{S}}^{(R_N)}]\bV[t_{\mathcal{S}}^{(R_N)}] \bbm[t_{\mathcal{S}}^{(R_N)}]}_{\doteq P}+\underbrace{\bh_k^{\He}[t_{\mathcal{S}}^{(R_N)}]\sum_{\ell\in\bar{\mathcal{S}}}\bv^{\ZF}_{\ell}[t_{\mathcal{S}}^{(R_N)}] a_{\ell}[t_{\mathcal{S}}^{(R_N)}]}_{\doteq P^{1-\alpha}}+\underbrace{z_k[t_{\mathcal{S}}^{(R_N)}]}_{\doteq P^{\alpha}}
\label{eq:Kuser_3}
\end{align}
where for $k\in \mathcal{K}$, we have 
\begin{align}
z_k[t]&\triangleq \underbrace{\bh^{\He}_k[t] \bv^{\ZF}_k[t]s_k[t]}_{\doteq P^{\alpha}}+\underbrace{\bh^{\He}_k[t]\sum_{\ell=1,\ell\neq k}^K \bv^{\ZF}_{\ell}[t]s_{\ell}[t]}_{\doteq P^{0}}.
\label{eq:Kuser_4}
\end{align}
At user $k\in\bar{\mathcal{S}}$, the received signal for $t=t_{\mathcal{S}}^{(R_N)}$ is given by
\begin{align}
y_k[t_{\mathcal{S}}^{(R_N)}]&=\underbrace{\bh_k^{\He}[t_{\mathcal{S}}^{(R_N)}]\bv^{\ZF}_{k}[t_{\mathcal{S}}^{(R_N)}] a_{k}[t_{\mathcal{S}}^{(R_N)}]}_{\doteq P}+\underbrace{i_k[t_{\mathcal{S}}^{(R_N)}]}_{\doteq P^{1-\alpha}}    +\underbrace{z_k[t_{\mathcal{S}}^{(R_N)}]}_{\doteq P^{\alpha}}
\label{eq:Kuser_5}
\end{align}
where we have introduced the short-hand notation $i_k[t_{\mathcal{S}}^{(R_N)}]$ for $k\in\bar{\mathcal{S}}$ as
\begin{align}
i_k[t_{\mathcal{S}}^{(R_N)}]&\triangleq \underbrace{\bh_k^{\He}[t_{\mathcal{S}}^{(R_N)}]\bV[t_{\mathcal{S}}^{(R_N)}] \bbm[t_{\mathcal{S}}^{(R_N)}]}_{\doteq P^{1-\alpha}}+\underbrace{\bh_k^{\He}[t_\mathcal{S}^{(R_N)}]\sum_{\ell\in\bar{\mathcal{S}},\ell\neq k}\bv^{\ZF}_{\ell}[t_{\mathcal{S}}^{(R_N)}] a_{\ell}[t_{\mathcal{S}}^{(R_N)}]}_{\doteq P^{1-\alpha}}.
\label{eq:Kuser_6}
\end{align}
For the above, it is easy to see (See for example~\cite{Jindal2006}) that ZF beamforming which uses $\alpha$-quality current CSIT, reduces the scaling of the interference power by a multiplicative factor of $P^{-\alpha}$.

\subsubsection{Generation of new data symbols}

The generation of the data symbols that are to be transmitted in phase~$j+1$ of round~$N$ and in phase~$j$ of round~$N+1$ from the interference generated during phase~$j$ of round~$N$ is one of the key ingredients of our scheme. We now consider that the transmissions of phase~$j$ have ended for every possible set~$\mathcal{S}\subset \mathcal{K}$ for which $|\mathcal{S}|=j$.

\paragraph{Preliminary step: Interference quantization}

At the end of phase~$j$, using delayed CSIT, the TX reconstructs $i_k[t_{\mathcal{S}}^{(R_N)}], \forall k\in\bar{\mathcal{S}}$, which have power that scales as $P^{1-\alpha}$. The next step depends on the value of $\alpha$.

\begin{itemize}
\item If $\alpha \leq \frac{1}{2}$, the TX uses the quantizer of Lemma~\ref{lemma_quantizer} to obtain
\begin{equation}
\hat{i}_k[t_{\mathcal{S}}^{(R_N)}]=Q_{1-\alpha,\alpha}(i_k[t_{\mathcal{S}}^{(R_N)}])
\label{eq:Kuser_7}
\end{equation}
that comes with a residual quantization noise $n_k[t_{\mathcal{S}}^{(R_N)}]$, which --- by design, and directly from the proof of Lemma~\ref{lemma_quantizer} --- has power scaling in $P^{\alpha}$.
Then the TX quantizes this quantization noise~$n_k[t_{\mathcal{S}}^{(R_N)}]$, to get $\hat{n}_k[t_{\mathcal{S}}^{(R_N)}]$, with $\alpha\log_2(P)$~bits, leaving a residual quantization noise that only scales in $P^0$ (cf.~\cite{Cover2006}).
Finally the TX combines (over the complex numbers) the two quantized estimates, to get a total estimate 
\begin{equation}
\hat{i}_k[t_{\mathcal{S}}^{(R_N)}]\triangleq Q_{1-\alpha,\alpha}(i_k[t_{\mathcal{S}}^{(R_N)}]) + \hat{n}_k[t_{\mathcal{S}}^{(R_N)}].
\label{eq:Kuser_8}
\end{equation}
\item If $\alpha \geq \frac{1}{2}$, the interference term $i_k[t_{\mathcal{S}}^{(R_N)}], \ k\in\bar{\mathcal{S}}$ is quantized using any $(1-\alpha)\log_2(P)$~bit quantizer, to directly give $\hat{i}_k[t_{\mathcal{S}}^{(R_N)}]$ with quantization noise that scales in $P^0$ (cf.\cite{Cover2006}). 
\end{itemize}
\paragraph{Generation of auxiliary data symbols for phase~$j$ of round $N+1$}
The above quantized estimates of the interference will be placed in \emph{auxiliary} data symbols to be transmitted during round~$N+1$. 
Depending on the value of $\alpha$, the auxiliary symbol $a_{k}[t_{\mathcal{S}}^{(R_{N+1})}]$ is loaded as follows
\begin{equation}
\begin{cases}
a_{k}[t_{\mathcal{S}}^{(R_{N+1})}]\longleftarrow\hat{n}_k[t_{\mathcal{S}}^{(R_N)}] &,\alpha\leq \frac{1}{2} \\
a_{k}[t_{\mathcal{S}}^{(R_{N+1})}]\longleftarrow\hat{i}_k[t_{\mathcal{S}}^{(R_N)}] &,\alpha\geq \frac{1}{2}
\end{cases}
\label{eq:Kuser_9}
\end{equation}
where the differentiation into two cases reflects that the rate of the auxiliary data symbol~$a_{k}[t_{\mathcal{S}}^{(R_N)}]$ will be set equal to $\min(1-\alpha,\alpha)\log_2(P)$~bits. We note that for the first round, all auxiliary data symbols are set to zero.
\begin{remark}
During each TS, the number of auxiliary data symbols generated matches the number of auxiliary data symbols transmitted, thus ensuring the proper functioning of the scheme. \qed
\end{remark}
\paragraph{Generation of $\QMAT$ data symbols for phase~$j+1$ of round $N$}
Let us consider an arbitrary set~$\mathcal{P}\subset \mathcal{K}$ with $|\mathcal{P}|=j+1$, and let us denote its elements as
\begin{equation}
\mathcal{P}\triangleq \{p_{1},\ldots,p_{j+1}\}.
\label{eq:Kuser_10}
\end{equation}
The $j$ order-$(j+1)$ data symbols\footnote{See further down Remark~$1$ for a clarification on having to repeat phases to accumulate enough symbols. This is done exactly as in MAT \cite{MaddahAli2012}, and it is transparent to the scheme and the performance here.} $m_{\ell}[t_{\mathcal{P}}^{(R_N)}],\ell\in\{1,\ldots,j\}$ \emph{for phase~$j+1$ of round $N$} are then defined as
\begin{align}
m_{\ell}[t_{\mathcal{P}}^{(R_N)}]&\triangleq \LB  \hat{i}_{p_{\ell}}[t_{\mathcal{P}\setminus {p_{\ell}}}^{(R_N)}] \oplus \hat{i}_{{p_{\ell+1}}}[t_{\mathcal{P}\setminus {p_{\ell+1}}}^{(R_N)}] \RB,\qquad \forall \ell\in\{1,\ldots,j\}.
\label{eq:Kuser_11}
\end{align}
Each of the above $2j$ components $\{\hat{i}_{p_{\ell}}[t_{\mathcal{P}\setminus {p_{\ell}}}^{(R_N)}], \ \hat{i}_{{p_{\ell+1}}}[t_{\mathcal{P}\setminus {p_{\ell+1}}}^{(R_N)}]\}_{\ell\in\{1,\ldots,j\}}$ has already been received at a RX, as some form of interference. In the next phase, the transmitter will recreate the $j$ different linear combinations $m_{\ell}, \ \ell\in\{1,\ldots,j\}$ and transmit them.
\begin{example}
For $K=3$, and $j=2$, this gives then~$\mathcal{P}=\{1,2,3\}$ and 
\begin{equation}
\begin{aligned}
m_{1}[t_{\{1,2,3\}}^{(R_N)}]&\triangleq (\hat{i}_1[t_{\{2,3\}}^{(R_N)}] &\oplus \hat{i}_2[t_{\{1,3\}}^{(R_N)}])&,\\
m_{2}[t_{\{1,2,3\}}^{(R_N)}]&\triangleq   & (\hat{i}_2[t_{\{1,3\}}^{(R_N)}] &\oplus \hat{i}_3[t_{\{1,2\}}^{(R_N)}]).
\label{eq:3user_encoding_15}
\end{aligned}
\end{equation}
\qed
\end{example}

\subsection{Phase~$K$}
Phase~$K$ is particular as it does not require to retransmit any data symbols. This is due to the fact that it is a broadcasting phase where a fully common message (meant for all users) is transmitted such that $\mathcal{S}=\mathcal{K}$ and $\bar{\mathcal{S}}=\emptyset$. During the single time-slot $t= t_{\mathcal{K}}^{(R_N)}$ of this phase, the transmitted signal takes the form
\begin{equation}
\xv[t]\!= \bv[t] m[t]+\sum_{k=1}^{K}\bv_{k}^{\ZF}[t]s_k[t]
\label{eq:Kuser_12}
\end{equation}
where
\begin{itemize}
\item $m[t] \in \mathbb{C}$ is a $\QMAT$ order-$K$ data symbol, having rate $(1-\alpha)\log_2(P)$~bits and power $\E\LSB|m[t]|^2\RSB\doteq P$, and $\bv[t]\in\mathbb{C}^K$ is a random unit-norm vector.
\item $s_k[t], \ k\in \mathcal{K}$ is a ZF data symbol meant for user~$k$, having rate $\alpha\log_2(P)$~bits and power $\E\LSB|s_k[t]|^2\RSB\doteq P^{\alpha}$.
\end{itemize}
During the TS~$t= t_{\mathcal{K}}^{(R_N)}$, user $k\in\mathcal{K}$ receives
\begin{align}
y_k[t]&=\underbrace{\bh_k^{\He}[t]\bv[t] m[t]}_{\doteq P}+\underbrace{z_k[t]}_{\doteq P^{\alpha}}.
\label{eq:Kuser_13}
\end{align}

\begin{remark}
It is important to note that vector $\bbm[t_{\mathcal{S}}^{(R_N)}]\in\mathbb{C}^{K-j}$ ($|\mathcal{S}|=j+1$) contains $K-j$ data symbols. To cover the gap from the fact that we have only generated $j$ order-$(j+1)$ data symbols $m_{\ell}[t_{\mathcal{P}}^{(R_N)}], \ell=1,\ldots,j$ ($|\mathcal{P}|=j$), one must simply repeat phase~$j$ exactly $\frac{K!}{j}$ times, as in the original MAT scheme\cite{MaddahAli2012}. This is automatically accounted for in the DoF calculation, as we will note later.
\qed
\end{remark}
\section{Scheme for the $K$-user case: Decoding and DoF}\label{se:K_user_decoding}

\subsection{Decoding: Proof by Induction}\label{se:K_user_decoding:induction}

We now consider decoding, and for simplicity assume that all transmissions of all phases and all rounds, have been completed\footnote{It will become clear that the data symbols of round~$N$ can be decoded as soon as round~$N+1$ has ended. Thus, the delay in the decoding of one data symbol does not increase with the number of rounds.}. The proof is done by induction.
Before making the inductive statement, we introduce the following set of decoded symbols
\begin{equation}
\mathcal{A}_{j,k}^{(R_{N})}\triangleq \left\{a_{\ell}[t_{\mathcal{S}}^{(R_N)}]  \ : \ \forall \mathcal{S}\subset\mathcal{K},\forall\ell \in \bar{\mathcal{S}}, |\mathcal{S}|=j,k\in\mathcal{S} \right\}.
\label{eq:Kuser_decoding_1}
\end{equation}
Intuitively, $\mathcal{A}_{j,k}^{(R_{N})}$ is the set of all the auxiliary data symbols generated by a transmission during phase~$j$ of round~$N$ where user~$k$ was among the destined users.

For any round~$N$ and phase~$j<K$, our induction statement is as follows:

If
\begin{itemize}
\item $(i)$ Phase~$j+1$ up to phase~$K$ of round~$N$ have been successfully decoded
\item $(ii)$ User~$k$ has decoded the set $\mathcal{A}_{j,k}^{(R_{N})}$
\end{itemize}
then
\begin{itemize}
\item $(i)$ User~$k$ can decode all its destined $\QMAT$ order-$j$ data symbols and its destined ZF data symbols transmitted during phase~$j$ of round~$N$
\item $(ii)$ User~$k$ can decode the set $\mathcal{A}_{j,k}^{(R_{N+1})}$.
\end{itemize}
After initializing the auxiliary data symbols of the first round to zero, each user $k$ can decode $\mathcal{A}_{j,k}^{(R_{1})}$ as they have a rate equal to zero. Thus we proceed with the inductive step for an arbitrary round~$N$ and phase~$j<K$, where by induction, it holds that all $\QMAT$ order-$(j+1)$ data symbols have been successfully decoded at the corresponding users, which means that user~$k$ has received the data symbols $m_{\ell}[t_{\mathcal{T}}^{(R_N)}],\ell=1,\ldots,j$, for any set $\mathcal{T}$ for which $|\mathcal{T}|=j+1$ and for which $k\in \mathcal{T}$.

In addition, still by induction, user~$k$ has decoded the set~$\mathcal{A}_{j,k}^{(R_{N})}$.

\paragraph{Decoding of the desired auxiliary data symbols}
As a first step, user~$k$ uses the signal received during round~$N+1$ to decode $a_k[t_{\mathcal{W}}^{(R_{N+1})}],\forall \mathcal{W}\subset \mathcal{K},|\mathcal{W}|=j,k\in \bar{\mathcal{W}}$. Indeed, it can be seen from~\eqref{eq:Kuser_5} that the scaling of the SINR matches the scaling of the rate. Again we differentiate between the cases $\alpha\geq \frac{1}{2}$ and $\alpha\leq \frac{1}{2}$ as the information contained in $a_k[t_{\mathcal{W}}^{(R_{N+1})}]$ which be different.
\begin{itemize}
\item If $\alpha\leq \frac{1}{2}$, user~$k$ has decoded $\hat{n}_k[t_{\mathcal{W}}^{(R_N)}]$ for every set $\mathcal{W}\subset\mathcal{K}$ for which $k\in \bar{\mathcal{W}}$ and for which $|\mathcal{W}|=j$ (see \eqref{eq:Kuser_9}). To recover the quantized interference $\hat{i}_k[t_{\mathcal{W}}^{(R_N)}]$, it is necessary for user~$k$ to obtain $  Q_{1-\alpha,\alpha}(i_k[t_{\mathcal{W}}^{(R_N)}])$  (see \eqref{eq:Kuser_8}), and this is achieved by quantizing the received signal~$y_k[t_{\mathcal{W}}^{(R_N)}]$ using the quantizer~$Q_{1-\alpha,\alpha}$. Indeed, it follows from Lemma~\ref{lemma_quantizer} that it holds with probability that approaches 1 as the transmit power~$P$ increases, that
\begin{equation}
\begin{aligned}
Q_{1-\alpha,\alpha}\LB y_k[t_{\mathcal{W}}^{(R_N)}]\RB&=Q_{1-\alpha,\alpha}\LB i_k[t_{\mathcal{W}}^{(R_N)}]\RB.
\end{aligned}
\label{eq:Kuser_decoding_2}
\end{equation}
\item If $\alpha\geq \frac{1}{2}$, user~$k$ has decoded directly from the auxiliary data symbols the quantized interference $\hat{i}_k[t_{\mathcal{W}}^{(R_N)}]$ for every set $\mathcal{W}$ with $|\mathcal{A}|=j$ such that $k\in \bar{\mathcal{W}}$ (see \eqref{eq:Kuser_9}).
\end{itemize}
In both cases, user~$k$ has now obtained~$\hat{i}_k[t_{\mathcal{W}}^{(R_N)}],\forall \mathcal{W}\subset\mathcal{K},|\mathcal{W}|=j,k\in \bar{\mathcal{W}}$.

\paragraph{Information at user~$k$ at the end of phase~$j+1$}
We now describe what are the data symbols available at user~$k$ after successfully decoding phase~$j+1$.

Note that for any set $\mathcal{T},|\mathcal{T}|=j+1$ for which $k\in \mathcal{T}$, we can write $\mathcal{T}=\{k,\mathcal{W}_{\mathcal{T},k}\}$ with $|\mathcal{W}_{\mathcal{T},k}|=j$ and $k\in \bar{\mathcal{W}}_{\mathcal{T},k}$. Thus from the decoding of the auxiliary data symbols during phase~$j$ in the previous paragraph, user~$k$ knows $\hat{i}_k[t_{\mathcal{W}_{\mathcal{T},k}}^{(R_N)}]$. Using these interference terms in combination with the order-$(j+1)$ data symbols $m_{\ell}[t_{\mathcal{T}}^{(R_N)}], \ell=1,\ldots, j$, user~$k$ is able to obtain all the quantized interference terms forming the order-$(j+1)$ data symbols, i.e., $\hat{i}_{\ell}[t_{\mathcal{W}_{\mathcal{T},\ell}}^{(R_N)}],\forall \ell \in \mathcal{T}$ with $\ell\neq k$ and $\mathcal{T}=\{\ell,W_{\mathcal{T},\ell}\}$, where again we note that $k\in \mathcal{W}_{\mathcal{T},\ell}$.

Considering all the decoded order-$(j+1)$ data symbols, user~$k$ has then decoded the set~$\mathcal{O}_k^j$ defined as
\begin{equation}
\mathcal{O}_k^j\triangleq\left\{\left\{\hat{i}_{\ell}[t_{\mathcal{T}\setminus \ell}^{(R_N)}]\right\}_{\ell\in\mathcal{T}} \ : \ \forall \mathcal{T}\subset\mathcal{K},|\mathcal{T}|=j+1, k\in\mathcal{T}\right\}.
\label{eq:Kuser_decoding_3}
\end{equation}

\paragraph{Decoding of the $\QMAT$ data symbols at user~$k$}
Let us now consider an arbitrary user~$k$ and an arbitrary set $\mathcal{S}\subset\mathcal{K},|\mathcal{S}|=j,k\in \mathcal{S}$. We will show that user~$k$ is able to decode its destined $K-j+1$~$\QMAT$ order-$j$ data symbols in $\bbm[t_{\mathcal{S}}^{(R_N)}]$. For that purpose, user~$k$ needs $K-j+1$ observations with a SINR scaling in $P^{1-\alpha}$. These observations will be:
\begin{align}
&y_k[t_{\mathcal{S}}^{(R_N)}]\\
&i_{\ell}[t_{\mathcal{S}}^{(R_N)}]\qquad,\forall \ell \in \bar{\mathcal{S}} .
\label{eq:Kuser_decoding_4}
\end{align}
Indeed, we can rewrite (up to the noise floor) the above quantized interference terms as $\hat{i}_{\ell}[t_{\mathcal{T}\setminus \ell}^{(R_N)}]$ with $\mathcal{T}=\{\mathcal{S},\ell\}$ and $k\in \mathcal{T}$. Thus, it can be seen from \eqref{eq:Kuser_decoding_4} that all the quantized interference terms~$\hat{i}_{\ell}[t_{\mathcal{S}}^{(R_N)}],\forall \ell \in \bar{\mathcal{S}} $ are part of $\mathcal{O}_k^j$. Hence, user~$k$ is able to obtain all the observations in \eqref{eq:Kuser_decoding_4}. Yet, the $\QMAT$ data symbols are interfered by the auxiliary data symbols $a_{\ell}[t_{\mathcal{S}}^{(R_N)}]$ (see \eqref{eq:Kuser_3} and \eqref{eq:Kuser_5}), but by induction, these data symbols are known at user~$k$, and thus their interference can be removed. Consequently, user~$k$ is able to decode its corresponding $\QMAT$ data symbols in $\bbm[t_{\mathcal{S}}^{(R_N)}]$.

\paragraph{Decoding of the ZF data symbols at user~$k$}
With the $\QMAT$ order-$j$ data symbols decoded, it is possible for user~$k$ to decode the ZF data symbols $s_k[t_{\mathcal{S}}^{(R_N)}]$ for all the time slots $t_{\mathcal{S}}^{(R_N)}$ for which $k\in \mathcal{S}$ ($|\mathcal{S}|=j$). When $k\in\bar{\mathcal{S}}$, user~$k$ cannot decode the $\QMAT$ data symbols transmitted, as it does not have sufficiently many observations. However, the user has decoded $\hat{i}_{\ell}[t_{\mathcal{S}}^{(R_N)}]$ for all sets $\mathcal{S}\subset \mathcal{K}$ with $|\mathcal{S}|=j$ and $k\in\bar{\mathcal{S}}$. Therefore it can remove the interference created by the $\QMAT$ data symbols up to the noise floor also for these TS, and then decode its destined ZF data symbol $s_k[t_{\mathcal{S}}^{(R_N)}]$.

\paragraph{Decoding of the auxiliary data symbols of round~$N+1$}
In order to conclude the inductive step, it remains to prove that it is possible for user~$k$ to decode the auxiliary data symbols $a_{\ell}[t_{\mathcal{S}}^{(R_{N+1})}],\forall\ell \in \bar{\mathcal{S}},|\mathcal{S}|=j,k\in \mathcal{S}$. This follows directly from the definition of the auxiliary data symbols in \eqref{eq:Kuser_9}. Indeed, these auxiliary data symbols at round~$N+1$ are a function of the auxiliary data symbols at round~$N$ (i.e., $a_{\ell}[t_{\mathcal{S}}^{(R_{N})}]$ for $\ell \in \bar{\mathcal{S}}$, which are assumed to be known by induction), and the $\QMAT$ data symbols in $\bbm[t_{\mathcal{S}}^{(R_N)}]$ which have been decoded at user~$k$. Therefore, user~$k$ is able to decode these auxiliary data symbols. This concludes the inductive step.

\subsection{Calculation of the DoF}\label{se:K_user_decoding:induction}
We first note that the last round is dedicated to transmitting only auxiliary data symbols, and that --- if the number of rounds is sufficiently large --- this DoF loss is negligible. Thus focusing on just one round, we note that we have followed closely the structure in different phases of the MAT scheme, so to calculate our DoF we first scale the MAT DoF by a factor of $(1-\alpha)$ to account for our reduced rate of the $\QMAT$ data symbols, and then note that during each TS, we additionally send $\alpha \log_2(P)$~bits to each user. Adding these two parts together, provides immediately the sum-DoF expression in Theorem~\ref{thm:mainTheorem}, and concludes the proof.

\section{Conclusion}
The work has provided the first ever communication scheme which, for the general $K$-user MISO BC, manages to simultaneously and optimally exploit both delayed and imperfect-quality current CSIT. This is achieved by providing a new way of jointly incorporating MAT-type alignment (based on delayed CSIT), and ZF-type separation (using imperfect-quality current CSIT). In addition to resolving a theoretical open problem, the $\QMAT$ scheme is designed to adapt to CSI considerations that span both timeliness and precision. Interestingly, the $\QMAT$ scheme can also be seen as a robust precoding scheme which exploits the knowledge of the past to make the transmission less dependent on the current channel state. Indeed, the DoF achieved with the $\QMAT$ scheme degrades more slowly than schemes from the literature as the CSIT quality coefficient $\alpha$ decreases. Investigating how to leverage the new ideas presented in this work to develop practical robust transmission schemes is an interesting research problem currently under investigation within our group.

\appendix
The quantizer $Q_{\beta_1,\beta_2}$ quantizes separately the imaginary part and the real part using half the bits for each dimension. Therefore, we consider in the following the quantization of a real signal~$y$. Let us start by introducing the short-hand notations:
\begin{equation}
\begin{aligned}
y_P&\triangleq \sqrt{P^{\beta_1}}y,\\
n_P&\triangleq \sqrt{P^{\beta_2}}n
\label{eq:appendix_1}
\end{aligned}
\end{equation}
and denote the points after applying the quantizer as
\begin{equation}
\begin{aligned}
\hat{y}_P&\triangleq Q_{\beta_1,\beta_2}(y_P)\\
\hat{y}^n_P&\triangleq Q_{\beta_1,\beta_2}(y_P+n_P).
\end{aligned}
\label{eq:appendix_2}
\end{equation}
We will now show how to design the quantizer~$Q_{\beta_1,\beta_2}$ such that
\begin{equation}
\lim_{P\rightarrow \infty}\Pr\{\hat{y}_P = \hat{y}^n_P\} = 1
\label{eq:appendix_3}
\end{equation}
and 
\begin{equation}
\E\LSB |\hat{y}_P-y_P|^2\RSB\dotleq P^{\beta_2}.
\label{eq:appendix_4}
\end{equation}
The quantizer $Q_{\beta_1,\beta_2}$ uses in fact only $\frac{(\beta_1-\beta_2)}{2}\log_2(P)-\frac{1}{2}\log_2(\log_2(P))$~bits. This unusual fact that a small fraction of the available bits are not used is a consequence of \eqref{eq:appendix_3}: The sensitivity of the quantization to an additive noise has to be controlled.

Building upon rate distortion theory\cite{Cover2006}, we compute the Mean Square Error (MSE)-minimizing optimal codebook with the rate $\frac{(\beta_1-\beta_2)}{2}\log_2(P)-\frac{1}{2}\log_2(\log_2(P))$~bits, which we denote by $\mathcal{C}^{\star}$. It is well known that the Gaussian source is the hardest to compress\cite{Cover2006}, such that we can upperbound the distortion achieved, denoted by~$D^{\star}$, by the distortion in the Gaussian case to obtain\cite{Cover2006}
\begin{align}
D^{\star}&\leq  P^{\beta_1} 2^{-2\LB\frac{(\beta_1-\beta_2)}{2}\log_2(P)-\frac{1}{2}\log_2(\log_2(P))\RB}\\
&=  \log_2(P)P^{\beta_2}\\
&\doteq P^{\beta_2}.
\end{align}
Using this quantizer allows then to satisfy \eqref{eq:appendix_4}. We will now modify this quantizer in order to also satisfy \eqref{eq:appendix_3}. Specifically, we modify~$\mathcal{C}^*$ such that the minimal distance between any two points is at least equal to~$\sqrt{\log_2(P)P^{\beta_2}}$. This is done as follows: If two points are closer than $\sqrt{\log_2(P)P^{\beta_2}}$, then we remove one of the two points of the codebook. Trivially, this does not increase the scaling of the MSE, such that \eqref{eq:appendix_4} remains valid. We denote this modified codebook by $\mathcal{C}_{\beta_1,\beta_2}$. The quantizer~$Q_{\beta_1,\beta_2}$ then maps the signal~$y$ to the nearest point inside $\mathcal{C}_{\beta_1,\beta_2}$ and we denote the set of the real values delimiting the quantization cells by~$\mathcal{B}_{\beta_1,\beta_2}$. 

We have now finalized the design of the quantizer and it remains solely to show that \eqref{eq:appendix_3} is indeed satisfied, To prove that result, we start by introducing a set~$\mathcal{A}$ containing all the points which are within~$\sqrt{\log(\log(P)) P^{\beta_2}}$ of the boundary of the quantization cell:
\begin{equation}
\mathcal{A} \triangleq \left \{x\in\mathbb{R}\big| \min_{q\in \mathcal{B}_{\beta_1,\beta_2}} |x-q| \leq \sqrt{\log(\log(P)) P^{\beta_2}}\right\}.
\end{equation} 
To gain insights, the sets~$\mathcal{C}_{\beta_1,\beta_2}$, $\mathcal{B}_{\beta_1,\beta_2}$, and $\mathcal{A}$ are illustrated in a toy-example in Fig.~\ref{ITW2016_journal_proof_lemma}. 
\begin{figure}
\centering
\includegraphics[width=0.7\columnwidth]{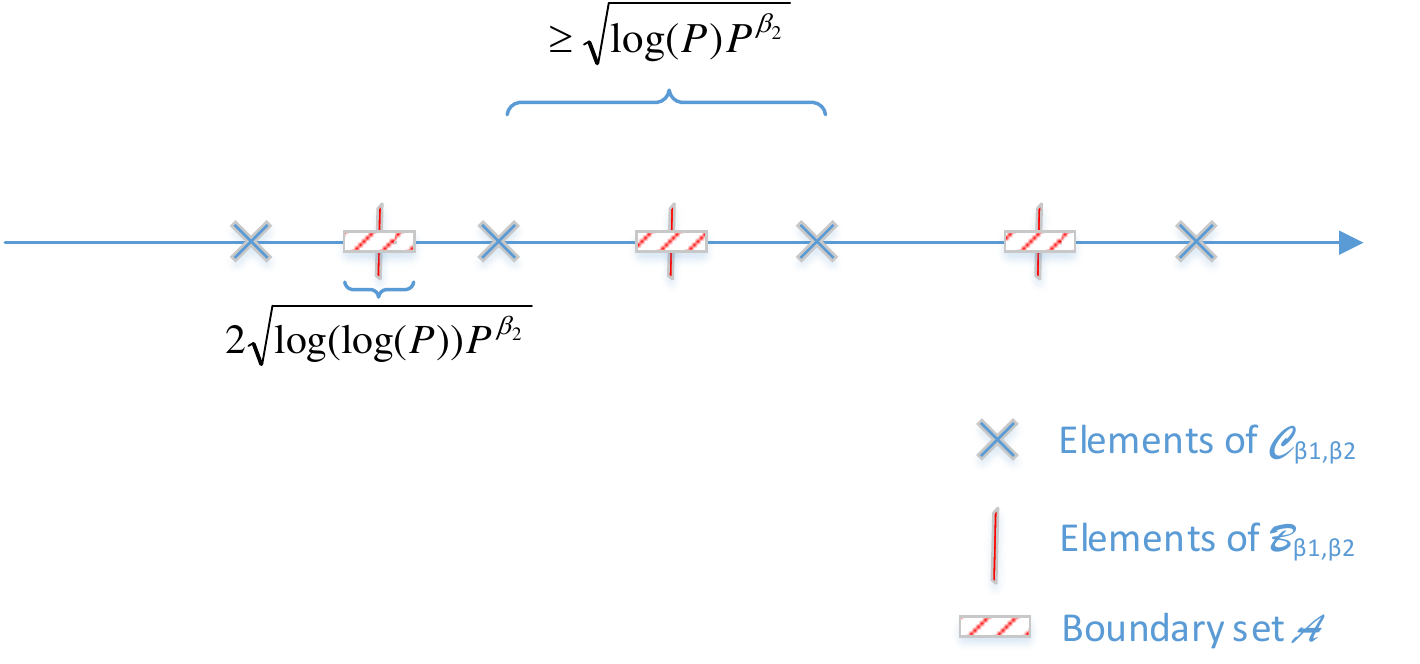}
\caption{Illustration of the quantizer design.} 
\label{ITW2016_journal_proof_lemma}
\end{figure}

We can then bound the probability that the additive noise~$n_P$ leads to a change of quantization point, i.e., $\hat{y}_P\neq \hat{y}^n_P$, as
		\begin{align}
			\Pr\left \{\hat{y}_P\neq \hat{y}^n_P \right\} &\!=\!\Pr\left\{\hat{y}_P\neq\hat{y}^n_P|y_P\in \mathcal{A}\right \}  \Pr\left\{ y_P\in\mathcal{A}\right\}\! +\! \Pr\left\{ \hat{y}_P\neq\hat{y}^n_P|y_P \notin \mathcal{A}\right \}\! \Pr\left \{y_P\notin \mathcal{A}\right\}\\
			&\leq \underbrace{\Pr\left\{ y_P\in\mathcal{A}\right\}}_{\triangleq P_1} + \underbrace{\Pr\left\{ \hat{y}_P\neq\hat{y}^n_P|y_P \notin \mathcal{A}\right \}}_{\triangleq P_2}
		\label{eq:probability_bound}
		\end{align} 
where we have denoted by~$P_1$ the first term of \eqref{eq:probability_bound} and by $P_2$ the second one. Focusing first on $P_2$, we can upperbound it as
\begin{align}
P_2&\leq \Pr\left\{|n_P|> \sqrt{\log(\log(P)) P^{\beta_2}}\right\}\\
&\stackrel{(a)}{\leq} \frac{P^{\beta_2}}{\log(\log(P))P^{\beta_2}}\\ 
&= \frac{1}{\log(\log(P))} 
\end{align}
where inequality~$(a)$ follows from Chebyshev's inequality.

Turning to $P_1$, we can write that
\begin{align}
P_1&= \Pr\left\{y_P\in \mathcal{A}\right\}\\
&= \Pr\left\{y\in \frac{\mathcal{A}}{\sqrt{P^{\beta_1}}}\right\}\\
&= \int_{\frac{\mathcal{A}}{\sqrt{P^{\beta_1}}}} P_{Y}(y)\mathrm{d}y \\
&\leq \sup_{y\in \mathbb{R}}p_Y(y) \int_{\frac{\mathcal{A}}{\sqrt{P^{\beta_1}}}}\mathrm{d}y.
\end{align} 
We can then upperbound the support of the set~$\mathcal{A}$ as its restriction to each quantization cell has the length of~$2\sqrt{\frac{\log_2(\log_2(P))P^{\beta_2}}{P^{\beta_1}}}$, for each of the $\sqrt{\frac{P^{\beta_1-\beta_2}}{\log_2(P)}}$ quantization cells. This gives
\begin{align}
P_1&\leq \sup_{y\in \mathbb{R}}p_Y(y) 2\sqrt{\frac{\log_2(\log_2(P))P^{\beta_2}}{P^{\beta_1}}}\sqrt{\frac{P^{\beta_1-\beta_2}}{\log_2(P)}}\\
&= \sup_{y\in \mathbb{R}}p_Y(y)  2\sqrt{\frac{\log_2(\log_2(P))}{\log_2(P)}}.
\end{align} 
Letting $P$ grow to infinity, both $P_1$ and $P_2$ tend to zero, which concludes the proof of the lemma.

\bibliographystyle{IEEEtran}
\bibliography{Literature}
\end{document}